\pdfmapfile{+FFF.map}
\documentclass[12pt,a4paper]{amsart}
\usepackage[english]{babel}
\usepackage{amsthm, amsmath, amssymb, amsfonts, geometry}
\usepackage[svgnames]{xcolor} 
\usepackage[colorlinks,citecolor=DarkGreen,urlcolor=DarkViolet, bookmarksdepth=0, linkcolor=Purple]{hyperref} 
\usepackage{kpfonts} 
\usepackage{graphicx}
\usepackage{tikz-cd}
\usepackage{indentfirst}
\usepackage{mathrsfs}
\usepackage{accents}
\usepackage{amsaddr}
\usepackage{xcolor}

\usepackage{listings}
\usepackage{color}

\usepackage{graphicx}
\usepackage{multirow}
\usepackage{mathrsfs}
\usepackage[title]{appendix}
\usepackage{xcolor}
\usepackage{textcomp}
\usepackage{manyfoot}

\usepackage{soul}
\usepackage[
    protrusion=true,
    activate={true,nocompatibility},
    final,
    tracking=true,
    kerning=true,
    spacing=true,
    factor=1100]{microtype}
\SetTracking{encoding={*}, shape=sc}{40}

\usepackage[
backend=biber,
style=alphabetic,
]{biblatex}
\newcommand{\done}[1]{}

\newtheorem{theorem}{Theorem}[section]

\newtheorem{example}[theorem]{Example}

\newtheorem{definition}[theorem]{Definition}
\newtheorem{definition-proposition}{Definition-Proposition}[section]
\newtheorem{definition-lemma}[theorem]{Definition-Lemma}
\newcommand{\lr}{\mathbin{\lrcorner}}
\usepackage[most]{tcolorbox}

\usepackage{geometry}
 \title{Structured Gaussians From Geometric Quantisation}
\author{Kerr Maxwell}
\address[A1]{School of Physics \& Astronomy, University of Birmingham, UK, B15 2TT}
\address[A1]{EPSRC CDT in Topological Design, University of Birmingham, UK, B15 2TT}
\email{kxm147@bham.ac.uk}
\date{\today}

\begin{document}

\begin{abstract}
We develop a geometric description of structured Gaussian beams, a form a structured light, by applying geometric quantisation and symplectic reduction to the 2D harmonic oscillator. Our results show that the geometric quantisation of the oscillator’s reduced phase space coincides with the modal Poincar\'e sphere in optics. We explicitly consider the case of the Generalised Hermite-Laguerre-Gauss modes, identifying their interbasis expansions with rotations of the reduced phase space and the geometric data accompanying the quantisation. This description simplifies the presentation of $SU(2)$ symmetries in structured light beams and is extensible to other symmetry groups.
\end{abstract}

\maketitle

\section{Introduction}\label{sec:introduction}
The study of scalar \textit{structured light} \cite{forbes2021structured} is concerned with understanding and manipulating families of solutions of the paraxial wave equation
\begin{equation}\label{eqn:pwe}
\left(\partial_x^2+\partial_y^2+2ik\partial_z\right)\psi = 0,
\end{equation}
where $\psi:\mathbb{R}^3\rightarrow \mathbb{C}$ describes a highly directional propagating monochromatic electromagnetic wave with wavenumber $k\in\mathbb{R}$. Such an equation is well-suited to describing radiation emitted by a laser, such as the fundamental Gaussian mode \cite{born2013principles}.
Different families of solutions to \eqref{eqn:pwe} - possessing a variety of optical properties - can be found, usually by separating \eqref{eqn:pwe} in different coordinate systems. Such solutions are generally called \textit{structured light modes} \cite{forbes2021structured} because their coordinate representations are typically expressed in terms of special functions.
Many distinct separable families of structured light modes have been studied; of particular importance are the Laguerre-Gaussian modes \cite{gbur2016singular} and Bessel Beams \cite{durnin1987diffraction}, as they carry orbital angular momentum \cite{allen1992orbital}, an additional dynamical degree of freedom not found in the fundamental Gaussian beam. Such new degrees of freedom and an improved understanding of how to control optical systems has seen structured light lead technological advances in areas such as optical manipulation \cite{simpson1997mechanical,lalaguna2024spatial} and trapping \cite{he1995optical}.

In this paper, we demonstrate how a particular 2-parameter family of structured light modes, the \textit{Generalised Hermite-Laguerre-Gauss} (GG) modes \cite{abramochkin2004generalized}, and their interbasis expansions admit a geometric interpretation. These modes are of particular interest as, while they reduce to the well-studied Hermite-Gaussian and Laguerre-Gaussian modes in special cases, generally they are not separable in any planar orthonormal coordinate system. In particular, we show (theorem \ref{thm:main}) how the interbasis expansions of GG modes in terms of the representation matrices of $SU(2)$ \cite{abramochkin2010generalized} may be computed from the geometric quantisation of the reduced phase space of the 2D harmonic oscillator. In doing so, the GG modes are found to be naturally parametrised by a collection of spheres equipped with some auxiliary geometric data defining a quantisation. Interbasis expansions of the modes are then identified with transformations of $S^2$ and the auxiliary geometric data accompanying the quantisation.

By interpreting \eqref{eqn:pwe} as a Schr\"odinger equation in $\hbar=1/k$, we can study the optical analogue of ``classical solutions" which may then be quantised to produce solutions to \eqref{eqn:pwe}. We are particularly interested in Gaussian solutions to \eqref{eqn:pwe}, defined as eigenfunctions of the 2D harmonic oscillator (2DSHO) Hamiltonian
\begin{equation}\label{eqn:2dsho}
    H=\frac{1}{2}\left(p_x^2+p_y^2+q_x^2+q_y^2\right),
\end{equation}
conceived as an operator in Fourier optics \cite{dennis2017swings}. In the optics literature, $H$ is identified with the \textit{beam quality factor} (usually denoted $M^2$) \cite{siegman1986lasers}. To obtain an explicit basis of Gaussian solutions to \eqref{eqn:pwe}, a second observable must also be supplied. Usually, this will be the separation constant of the coordinate system in question, the emergence of angular momentum in cylindrical coordinates being the most obvious example. To describe the 2-parameter family of $GG$ modes, we encode the second observable as geometric information 
 in the quantisation of the reduced phase space associated with $H$. Finally, we will identify this reduced space with the \textit{modal Poincar\'e sphere} \cite{padgett1999poincare,dennis2017swings,alonso2017ray,dennis2019gaussian} that has seen much recent attention in structured light (see \cite{he2022towards} for a review, including a discussion of the related \textit{modal Majorana sphere} \cite{gutierrez2020modal}).

In section \ref{sec:background} we summarise the background material on symplectic reduction and geometric quantisation. In section \ref{sec:js} we perform the geometric quantisation of \eqref{eqn:2dsho} over its original phase space, then the reduced phase space in section \ref{sec:redQuant}. Finally, in section \ref{sec:moderotate} we derive the interbasis expansions of the previous quantisations and show they agree with those of the GG modes. While our primary aim is to prove theorem \ref{thm:main}, we also try to show our calculations in sufficient detail that physicists -- so inclined -- may adapt these ideas to other systems.

\section{Background}\label{sec:background}

\subsection{Symplectic Reduction}\label{sec:symplectic_reduction}
Let $(M,\omega)$ be a symplectic manifold. Symplectic reduction studies the structure-preserving quotients of $M$ by the action of a Lie group $G$. This operation becomes particularly interesting when $G$ acts in a \textit{Hamiltonian} fashion, giving the quotient manifold, also called the \textit{reduced phase space}, the interpretation of the moduli space of $G$-orbits. The original formulation of symplectic reduction is due to Marsden and Weinstein \cite{marsden1974reduction} and can be found in many standard reference works (see \cite{mcduff2017introduction} or \cite{abraham1978foundations, marsden2013introduction}).

\begin{definition}
    Let $(M,\omega)$ be a symplectic manifold and $\mathfrak{X}(M)$ be the set of smooth vector fields on $M$. We say $X\in\mathfrak{X}(M)$ is a symplectic vector field if $X \lr\omega$ is closed. Given a smooth function $H:M\rightarrow \mathbb{R}$ we call $X_H\in\mathfrak{X}(M)$ a Hamiltonian vector field with Hamiltonian $H$ iff it satisfies 
    \begin{equation}
        X\lr\omega+dH=0.
    \end{equation}
\end{definition}

Note that different sign conventions for the definition of $X_H$ are used in the literature. Given a Hamiltonian $H$, the integral curves of $X_H$ are the solutions of Hamilton equations $X_H(z) = \dot{z}$ for $z\in M$. In principle, these equations may be integrated to find an associated $G$-action. To go the other way, we need to consider the Lie algebra.
\begin{definition}
    Let $\Phi_g:G\times M\rightarrow M$ be a group action parametrised by a group element $g\in G$. Let $\xi\in\mathfrak{g}$ be an element of the Lie algebra $\mathfrak{g}$ of $G$ such that $e^{t\xi}=g$ for $t\in\mathbb{R}$. The action $\Phi$ is called Hamiltonian if the generator
    \begin{equation*}
        X_\xi : = \frac{d}{dt}\Big|_{t=0}\Phi_{e^{t\xi}}.
    \end{equation*}
    is a Hamiltonian vector field and the action is equivariant with respect to the adjoint action.
\end{definition}

Given a Hamiltonian action of $G$ on $M$, its flow gives rise to constants of the motion, described using moment maps. 

\begin{definition}
    Let $\Phi$ be a Hamiltonian action of $G$ on $M$ and choose a pairing $\langle\_,\_\rangle:\mathfrak{g}^*\times\mathfrak{g}\rightarrow \mathbb{R}$. A moment map $\mu$ for $\Phi$ is a $G$-equivariant smooth map $\mu:M\rightarrow\mathfrak{g}^*$ which satisfies
    \begin{equation*}
        H_\xi(m) = \langle\mu(m),\xi \rangle
    \end{equation*}
    at every point $m\in M$.
\end{definition}
In particular, one finds there is a Lie algebra homomorphism from $\mathfrak{g}$ to the Poisson algebra of smooth functions on $M$ via $\langle\mu(m),[\xi,\eta] \rangle = \omega(X_\xi,X_\eta)(m)$ for $\xi,\eta\in M$. We now have enough Hamiltonian structure developed to describe reduction.

\begin{theorem}[Marsden-Weinstein Reduction, 5.4.15 in \cite{mcduff2017introduction}]
	Let $\Phi$ be a free and proper Hamiltonian $G$ action on $M$ with moment map $\mu$. Let $\Phi$ act freely and properly on $\mu^{-1}(\mathcal{O})$ where $\mathcal{O}\subset \mathfrak{g}^*$ is a coadjoint orbit and every point in $\mathcal{O}$ is a regular value of $\mu$, i.e. $\mu^{-1}(\mathcal{O})$ is a submanifold. Then the quotient
	\begin{equation*}
		M_{\mathcal{O}}:=\mu^{-1}(\mathcal{O})/G
	\end{equation*}
	is a symplectic manifold, called the reduced phase space. The symplectic structure $\overline{\omega}$ on $M_\mathcal{O}$ is uniquely given by restriction and inclusion (see the \textit{shifting trick} in the literature)
	\begin{equation*}
		\overline{\omega}_{[m]}([v_1],[v_2]): = \omega(v_1,v_2) - \langle\mu(m),[\xi_1,\xi_2]\rangle
	\end{equation*}
    for $v_i\in T_m\mu^{-1}(\mathcal{O})$, where $\xi_i$ are chosen such that they satisfy the shifted moment condition
    \begin{equation*}
        d\mu(p)(v_i-X_{\xi_i})=0.
    \end{equation*}
\end{theorem}

This reduction procedure provides a powerful way to study Hamiltonian systems. To summarise, when a phase space (symplectic manifold) is acted upon by a group in such a way that it admits a moment map with the necessary technical properties, then the structure of Hamilton's equations is preserved by the natural quotient with respect to the group action on the level sets of the moment map. While the resulting reduced phase space will generally be more complicated topologically, the reduction in dimension significantly simplifies the dynamics. Some explicit examples of moment maps are discussed in section \ref{sec:qviar}.

\subsection{Geometric Quantisation}\label{sec:geometric_quantisation}
Originated by Kostant \cite{kostant1970quantization} and Souriau \cite{souriataur1966quantification}, the goal of geometric quantisation is to extend the canonical quantisation prescription 
\begin{equation*}
    q\mapsto \widehat{q}\,,\quad p\mapsto-i\hbar\frac{\partial}{\partial q}\,,
\end{equation*}
on $(q,p)\in T^*\mathbb{R}^n$ to more general phase spaces. The starting point for this is the observation that $H\mapsto X_H$ is a Lie algebra homomorphism
\begin{equation*}
    [X_f,X_g]=X_{\{f,g\}}.
\end{equation*}
By building a special complex line bundle $\mathbb{L}$ over $M$, this relationship between Hamiltonians and their vector fields can be extended to covariant derivatives of sections $s\in\Gamma(\mathbb{L})$. Given the prescription of some additional geometric data, a Hilbert space $\mathcal{H}$ of sections may be formed along with a collection of operators $\mathcal{O}(\mathcal{H})$ on $\mathcal{H}$ which satisfactory model quantisation. The main reference work for geometric quantisation is \cite{woodhouse1992geometric}, see also \cite{hall2013quantum}. A general discussion of quantisation techniques and the \textit{Dirac Axioms} as a model of quantisation is given in \cite{ali2005quantization}.

\begin{definition}\label{def:prequant}
    Let $(M,\omega)$ be a $2n$-dimensional symplectic manifold. A prequantum line bundle over $M$ is a Hermitian complex line bundle $\mathbb{C}\hookrightarrow \mathbb{L}\rightarrow M$ with compatible Hermitian connection $\nabla$ with $\text{curv}(\nabla)=\omega$.
\end{definition}

The Hermitian condition on $\mathbb{L}$ can be thought of as requiring $\omega$ to represent an integral cohomology class, or equivalently, that $c_1(\mathbb{L})=\frac{N}{2\pi\hbar}\omega$ for $N\in\mathbb{Z}$. If a particular $M$ supports a prequantum structure, we say $M$ is \textit{prequantisable}. Given a prequantisation of $M$, square-integrable sections of $\mathbb{L}$ form a natural Hilbert space, the \textit{prequantum Hilbert space}
\begin{equation*}
    \mathcal{H}_{pre} = \left\{ s\in\Gamma(\mathbb{L})\;:\;||s||<\infty\right\},
\end{equation*}
where
\begin{equation*}
||s||^2 = \int_M(s,s)\,\epsilon_\omega,
\end{equation*}
is given by the integral of the Hermitian inner product in each fibre against a conventionally scaled volume form $\epsilon_\omega \propto \omega^n$. 
Classical observables (smooth functions on $M$) can be lifted to operators on $\mathcal{H}_{\text{pre}}$ using the \textit{prequantum operator} prescription
\begin{definition}\label{def:prequanOperator}
    Let $(M,\omega)$ be prequantisable. Then the ``hat" map \begin{equation}
        \widehat{}\,\,\,:C^\infty(M)\rightarrow \mathcal{O}(\mathcal{H}_{\text{pre}}),
    \end{equation}
    given by
    \begin{equation*}
      f\longmapsto\widehat{f} \;= -i\hbar \nabla_{X_f}+f 
    \end{equation*}
    produces operators on $\mathcal{H}_{\text{pre}}$ satisfying the Dirac quantisation axioms, in particular, it maps Poisson brackets to commutators 
    \begin{equation*}
        [\widehat{f}_1,\widehat{f}_2] = -i \hbar \widehat{\left\{ f_1,f_2\right\}}.
    \end{equation*}
    Self-adjointness of $\widehat{f}$ follows from the completeness of $X_f$.
\end{definition}

To write $\widehat{f}$ in local coordinates, choose a local trivialisation of $\mathbb{L}$ by choosing a connection form $\theta$. Since $\text{curv}(\nabla)=\omega$, the connection form must be a symplectic potential, satisfying $d\theta=\omega$. Locally we have
\begin{equation}
    \widehat{f} = -i\hbar\left(X_f-\frac{i}{\hbar}\theta(X_f)\right)+f.
\end{equation}

From a physical perspective $\mathcal{H}_\text{pre}$ is not a suitable Hilbert space for a quantum-mechanical system, as its elements generally depend on position \textit{and} momentum, in violation of the Uncertainty Principle. An extra piece of geometric data called a \textit{polarisation} is introduced to ensure the dependence of the elements of a Hilbert space formed from $\Gamma(\mathbb{L})$ on local coordinates is physically admissible. From here it becomes desirable to work with the complexification of certain structures, where this is the case we denote complex conjugation of an object $Z$ by $\overline{Z}$.

\begin{definition}
    A polarisation of a symplectic manifold $M$ is an integrable Lagrangian distribution $P \subset T^{\mathbb{C}}M$ which is at each point closed under the Jacobi-Lie bracket and has constant $\dim(P_m\cap\overline{P}_m)$.
\end{definition}

If $\overline{P}=P$ or $P\cap\overline{P}=\emptyset$ we call $P$ a \textit{real} or \textit{purely complex} polarisation respectively.
\begin{example}
    For any smooth manifold $N$, the symplectic manifold $M=T^*N$ has two natural real polarisations; The vertical polarisation defined by $P = \mathrm{span}(\partial_p)\subset T^\mathbb{C}M$ and the horizontal polarisation defined by $P = \mathrm{span}(\partial_q)\subset T^\mathbb{C}M$. 
\end{example}
The best-behaved polarisations are generally those generated from K\"ahler structures on $M$.
\begin{definition}
    Let $M$ be a symplectic manifold with a purely complex polarisation $P$. Let $J$ be the unique almost complex structure where $J|_P = i\;\mathbb{I}$ and $J|_{\overline{P}} = -i\;\mathbb{I}$. $P$ is called a K\"ahler polarisation if the Hermitian metric 
    \begin{equation*}
        g(X,Y) = \omega(X, JY)
    \end{equation*}
    is positive definite for all $X, Y\in T^\mathbb{C}M$.
\end{definition}
\begin{example}
    Consider $T^*\mathbb{R}^n$ with the canonical symplectic structure. Under the complex coordinates $z_i=q_i-ip_i$ the symplectic structure takes the form $\omega= -\frac{i}{2}dz_i\wedge d\bar{z}_i$. The span of $\partial_{z}$ at each point defines a K\"ahler polarisation.
\end{example}

The physical admissibility of sections is then defined with respect to dependence on a particular polarisation.
\begin{definition}
    Let $\mathbb{L}\rightarrow M$ be a prequantisation with connection $\nabla$ and $P$ a polarisation on $M$. A section $s\in\Gamma(\mathbb{L})$ is said to be polarised iff
    \begin{equation*}
        \nabla_Xs=0\quad\text{for all }X\in \overline{P}.
    \end{equation*}
\end{definition}
From these polarised sections we can build a new Hilbert space $\mathcal{H}_P$, called the \textit{polarised} or \textit{quantum} Hilbert space.
\begin{definition}
    The quantum Hilbert space associated to a prequantisation $\mathbb{L}\rightarrow M$ with connection $\nabla$ and polarisation $P$ is
\begin{equation*}
\mathcal{H}_{P}(M) = \left\{s\in\Gamma(\mathbb{L})\;|\; \nabla_{X}s=0\;\text{for all}\;\;X\in \overline{P}\;\text{and}\;||s||<\infty\right\}.
\end{equation*}
We will often drop the $M$ it is clear from context.
    \end{definition}

It is possible for $\mathcal{H}_P$ to be empty, especially if $P$ is a real polarisation of a noncompact manifold, for example, the horizontal or vertical polarisation of $T^*\mathbb{R}$. In this case an extra piece of data, a bundle of $\frac{1}{2}$-forms, is needed to define a sensible volume from on $D:=M/P$ from which a more useful notion of square-integrability can be constructed. Recall that the canonical bundle of a rank $k$ vector bundle $E$ is given by the rank 1 product bundle $\mathcal{K}(E)=\bigwedge^kE^*$, whose sections are top forms. Given a manifold $M$, we define the canonical bundle over $M$ as $\mathcal{K}(M) := \mathcal{K}(T^\mathbb{C}M)$.
\begin{definition}
    Let $E$ be a vector bundle. A square root of $E$ is a vector bundle, denoted $\sqrt{E}$, along with a fixed isomorphism $\sqrt{E}\otimes\sqrt{E}\cong E$. In other words, it is the formal square root (up to isomorphism) with respect to the tensor product. Sections of $\sqrt{E}$ are called $\frac{1}{2}$-forms. 
\end{definition}
In general, there is a topological obstruction encountered when defining a square-root bundle, see \cite{atiyah1971riemann} for a comprehensive treatment.

\begin{definition}\label{def:correctedHilb}
    Let $P$ be a polarisation of $M$ such that $M/P$ is a manifold. Let $\mathbb{L}\otimes \delta_P\rightarrow M$ be a prequantisation of $M$ with $\frac{1}{2}$-form correction, connection $\nabla$ and polarisation $P$. The corrected quantum Hilbert space is
    \begin{equation*}   
    \widetilde{\mathcal{H}}_{P}(M) = \left\{\widetilde{s}=s\otimes\nu\in\Gamma(\mathbb{L}\otimes\delta_P)\;|\; \nabla_{X}s=0\;\text{for all}\;\;X\in \overline{P}\;\text{and}\;||\;\widetilde{s}\;||<\infty\right\},
     \end{equation*}
     where
     \begin{equation*}
         ||\widetilde{s}||:=\int_{M/P}(s,s)\,(\nu,\nu)
     \end{equation*}
     is given by
     \begin{equation*}
         (\nu,\nu) = \sqrt{\frac{\overline{(\nu\otimes\nu)}\wedge(\nu\otimes\nu)}{\epsilon_\omega}}\,\,.
     \end{equation*}
\end{definition}

The prequantum operator extends to sections $\widetilde{s}\in\Gamma(\mathbb{L}\otimes\delta_P)$, we denote an extension of $\widehat{f}$ to operate on $\delta_P$ by $\widetilde{f}$. 

\begin{definition}
    Let $\mathbb{L}\otimes \delta_P\rightarrow M$ be a prequantisation of $M$  with $\frac{1}{2}$-form correction, with respect to a polarisation $P$. 
    The $\frac{1}{2}$-form corrected pre-quantum operator associated to $f\in C^\infty(M)$ acting on a section $\widetilde{s}=s\otimes\nu\in\Gamma\left(\mathbb{L}\otimes \delta_P\right)$ is given by
    \begin{equation*}
        \widetilde{f}(\;\widetilde{s}\;) = \widehat{f}\otimes\nu -i\hbar s\otimes \mathcal{L}_{X_f}(\nu).
    \end{equation*}
\end{definition}
The Lie derivative of a $\frac{1}{2}$-form may be computed using the Leibniz rule $\mathcal{L}_X(A\otimes B) = (\mathcal{L}_XA)\otimes B + A\otimes( \mathcal{L}_X B)$ and Cartan's homotopy formula.

Given a (possibly $\frac{1}{2}$-form corrected) quantisation of $M$, the polarisation chosen limits the space of quantisable observables. Admissible observables should map polarised sections to polarised sections, i.e. $\mathcal{H}_P$ should be closed under the action of quantum operators. It turn out that $\widehat{f}$ preserves polarised sections iff $[X_f,Y]\in P$ for all $Y\in P$. In the case of $\mathcal{H}_{\partial_p}(T^*\mathbb{R})$ this means the only quantisable functions are those of the form $a(q)p+b(q)$, such observables are called \textit{directly quantisable}. It is possible to quantise observables with up to quadratic dependence on canonical momentum, such as kinetic energy in the momentum polarisation, by use of the \textit{BKS pairing}, a special map between quantum Hilbert spaces due to Blattner, Kostant and Sternberg \cite{blattner2006metalinear}. In general, two geometric quantisations of the same space but with respect to different polarisation will give rise to different Hilbert spaces which may not be unitarily related. This will not be a problem for us, as all our Hilbert spaces are related through the $B_N$.

\section{K\"ahler Quantisation of the 2DSHO and the Jordan-Schwinger Map}\label{sec:js}

In this, section we see how geometric quantisation and symplectic reduction can be combined in a simple way to quantise the 2DSHO. This section serves as a warm-up for section \ref{sec:redQuant} and an opportunity to see the machinery introduced in section \ref{sec:background} working in a familiar setting. First, we demonstrate the explicit geometric quantisation of flat phase space with respect to \eqref{eqn:2dsho}, constructing a $\frac{1}{2}$-form corrected Hilbert space, then we use symplectic reduction to show how the same can be achieved by ``factoring" a reduced phase space.
The interplay of quantisation and reduction itself, originally studied in \cite{guillemin1982geometric}, has a rich history.

\subsection{Direct Quantisation}
Recall that the Jordan-Schwinger map \cite{jordan1935zusammenhang,schwinger1952angular} takes the Hilbert space of a quantum mechanical system into the $N$-dimensional Segal-Bargmann space $B_N$ of the product of $N$ quantum harmonic oscillators. Following Hall \cite{hall2013quantum} we define $B_N$ as the space of holomorphic functions $F$ in $N$ complex variables $\mathbf{z}$ which are finite with respect to the norm
\begin{equation}\label{eqn:bargnorm}
    ||F||^2_{B_N} = \frac{1}{(\pi\hbar)^N}\int|F(\mathbf{z})|^2 e^{-\frac{1}{\hbar}|\mathbf{z}|^2}d\mathbf{z}.
\end{equation}

For simplicity we consider the explicit quantisation of the 2D harmonic oscillator on flat phase space $M=T^*\mathbb{R}^2$, though the argument may be straightforwardly extended to harmonic oscillators in any dimension. Take the canonical K\"ahler structure on $T^*\mathbb{R}^2$ with complex coordinates $z_j=q_j-ip_j$. Let $\omega = i\partial\bar{\partial}\mathcal{K}$ for $\mathcal{K}=\frac{1}{2}|\mathbf{z}|^2$ a K\"ahler potential.
A particularly nice local trivialisation is associated to the connection form
\begin{equation*}
	\theta = \frac{i}{2}\left(\bar{\partial}\mathcal{K}-\partial\mathcal{K}\right).
\end{equation*}
We will use the holomorphic polarisation, which is a K\"ahler polarisation $P=\text{span}(\partial_{\mathbf{z}})$. Admissible sections $\psi\in\mathcal{H}_{\partial_\mathbf{z}}$ are those satisfying the polarisation condition
\begin{equation*}
    \nabla_{\partial_{\bar{z}_i}}\psi = 0\quad\text{for}\;\; i=1,2.
\end{equation*}
Admissible sections thus take the form
\begin{equation}\label{eqn:kahlertemp}
    \psi(\mathbf{z},\bar{\mathbf{z}}) = F(\mathbf{z})e^{-\frac{1}{4\hbar}|\mathbf{z}|^2},
\end{equation}
where $F(\mathbf{z})$ is a holomorphic function of $z_1$ and $z_2$. Taking the Hermitian inner product to be the standard inner product on each $\mathbb{C}$ fibre, we have
\begin{equation*}
    ||\psi||^2 = \int |F(\mathbf{z})|^2 e^{-\frac{1}{2\hbar}|\mathbf{z}|^2} \epsilon_\omega,
\end{equation*}
where
\begin{equation*}
    \epsilon_\omega = \frac{1}{\left(2\pi\hbar\right)^2}\frac{\omega^2}{2!}
\end{equation*}
is our conventional volume form on $M$. $\mathcal{H}_{\partial_\mathbf{z}}$ is given by sections of the form \eqref{eqn:kahlertemp} subject to the above norm. To apply the $\frac{1}{2}$-form correction to $\mathcal{H}_{\partial_\mathbf{z}}$, observe that the canonical bundle associated with the polarisation is trivial and the square root $\delta_{\partial_\mathbf{z}}$ has unit section $\sqrt{dz_x\wedge dz_y} = \sqrt{d\mathbf{z}}$. The corrected space $\widetilde{\mathcal{H}}_{\partial_\mathbf{z}}$ follows from definition \ref{def:correctedHilb}.

Let $A$ be a classical observable such that $X_A$ preserves the polarisation $\partial_{\mathbf{z}}$. For completeness, we give the general form of $\widetilde{A}$, the quantum operators with half-form correction
\begin{equation}
\begin{split}
	\tilde{A}\tilde{\psi} = \Bigg(-\frac{i\hbar}{\omega}\left[\frac{\partial A}{\partial z_i} \frac{\partial \psi}{\partial \bar{z}_i}- \frac{\partial A}{\partial \bar{z}_i} \frac{\partial \psi}{\partial z_i}\right]+\left[-\frac{1}{2\omega}\left(\frac{\partial A}{\partial z_i} \frac{\partial \mathcal{K}}{\partial \bar{z}_i}- \frac{\partial A}{\partial \bar{z}_i} \frac{\partial \mathcal{K}}{\partial z_i} \right)+A\right]\psi\Bigg)\otimes\sqrt{d\mathbf{z}}\quad\cdots\\
    -\frac{i\hbar}{\omega}\psi\otimes L_{X_A}\sqrt{d\mathbf{z}},
    \end{split}
\end{equation}
where we are implicitly summing over $i=x,y$.
In the case of $\widetilde{H}$ acting on $\widetilde{\mathcal{H}}_{\partial_\mathbf{z}}$ we have 
\begin{equation*}
    \widetilde{H}\widetilde{\psi} = \hbar\left(z_y \frac{\partial F}{\partial z_y}+z_x\frac{\partial F}{\partial z_x}+F\right)e^{-\frac{1}{4}|\mathbf{z}|^2}\otimes\sqrt{d\mathbf{z}},
\end{equation*}
which admits a nice basis of eigenfunctions $\psi_{n_x,n_y}$ in terms of polynomials of the $z_i$,
\begin{align*}
    \widetilde{H}\widetilde{\psi}_{n_x,n_y} &=\widetilde{H}\left(z_x^{n_x}z_y^{n_y}e^{-\frac{1}{4}|\mathbf{z}|^2}\otimes\sqrt{d\mathbf{z}}\right)\\    
    &=\hbar\left(n_x+n_y+1\right)\widetilde{\psi}_{n_x,n_y}.
\end{align*}
The rest of the quantisable complex functions give $\widetilde{z_i} = z_i$ and $\widetilde{\bar{z}_i} = \partial_{z_i}$; our creation and annihilation operators respectively, for each oscillator degree of freedom. The simplest way to understand this Hilbert space is by mapping it into the Segal-Bargmann space via a unitary (up to overall constant determined by $vol(M)$) map $\varphi:\widetilde{\mathcal{H}}_{\partial_{\bar{z}}} \rightarrow B_2$, where a unitary basis with respect to \eqref{eqn:bargnorm} is given by
\begin{equation}\label{eqn:sBasis}
\varphi\left( \frac{\widetilde{\psi}_{n_x,n_y}}{\sqrt{n_x! \;n_y!}}\right) = |n_x,n_y\rangle = \frac{\left(a_y^\dag\right)^{n_y}}{\sqrt{n_y!}}\frac{\left(a_x^\dag\right)^{n_x}}{\sqrt{n_x!}}|0\rangle.
\end{equation}
for $a_i^\dag = \varphi \widetilde{\bar{z}}_i\varphi^{-1}$ and $a_i = \varphi \widetilde{\bar{z}}_i\varphi^{-1}$ the creation and annihilation operators and $|0\rangle$ the ground state in $B_2$. For clarity, we write states in $B_N$ using \textit{bra-ket} notation, where kets are identified with unitary basis elements. In this basis the $a_i$s and $a^\dag_i$s act on each component individually like
\begin{align}
    a^\dag_x|n_x,n_y\rangle =&\; \sqrt{n_x+1}|n_x+1,n_y\rangle,\\
    a_x|n_x,n_y\rangle =&\; \sqrt{n_x}|n_x-1,n_y\rangle,
\end{align}
and similarly on the $y$ component. This is the \textit{Schwinger representation} of the 2D harmonic oscillator. 

\subsection{Quantisation via Reduction}\label{sec:qviar}
Now we can perform the same calculation via a different route. Given that we have the splitting $T^*\mathbb{R}^2\cong T^*\mathbb{R}\times T^*\mathbb{R}$, we can construct $\widetilde{\mathcal{H}}_{\partial_{\mathbf{z}}}$ by using the associated splitting of the quantisation.
We seek a symplectic reduction $T^*\mathbb{R}^2\rightarrow T^*\mathbb{R}$ that manifests an example of the famous theorem by Guillemin and Sternberg (see \cite{guillemin1982geometric} for the original theorem and \cite{guillemin2002moment} for an updated discussion), which states, roughly speaking, that symplectic reduction and geometric quantisation are commuting operations.

Consider the Hamiltonian $\mathbb{R}$-action of $r\in\mathbb{R}$ on $(q_x,p_x,q_y,p_y)\in T^*\mathbb{R}^2$
\begin{equation*}
\Phi_r(q_x,p_x,q_y,p_y)\mapsto (q_x,p_x,q_y+r,p_y).
\end{equation*}
Taking $r=e^{t\xi}$ for $t,\xi\in\mathbb{R}$, such an action is generated by a vector field $X_\xi$ where $\xi$ is an element of the Lie algebra of $\mathbb{R}$. The generator is
\begin{equation*}
	X_\xi = \frac{d}{dt}\Big|_{t=0}\Phi_{e^{t\xi}}
\end{equation*}
We have that $X_\xi(q_y) = \xi$ and so $X_\xi = \xi\partial_{q_y}$. $X_{\xi}$ is associated with the Hamiltonian function
\begin{eqnarray}
	H_\xi:T^*\mathbb{R}^2\rightarrow \mathbb{R}:(q_x,p_x,q_y,p_y)\mapsto\xi \,p_y.
\end{eqnarray}
Given this the moment map associated to the action is the $G$-equivariant map $\mu:T^*\mathbb{R}^2\rightarrow \mathbb{R}$ satisfying
\begin{equation*}
	H_\xi(q_x,p_x,q_y,p_y) = \langle \mu(q_x,p_x,q_y,p_y),\xi\rangle,
\end{equation*}
which is clearly given by the projection $\mu(q_x,p_x,q_y,p_y) = p_y$, where the momenta now take values in the dual of the Lie algebra. The equivariance of $\mu$ follows from the triviality of the $\mathrm{Ad}^*$ action. The reduced phase space associated to a point $p_y$ is $H_\xi^{-1}(p_y)/\Phi\cong T^*\mathbb{R}$ with the standard symplectic structure. Let $\mathcal{H}_{\partial_{z}}(H_\xi^{-1}(p_y)/\Phi)$ denote the Hilbert space  obtained from the K\"ahler quantisation of $H_\xi^{-1}(p_y)/\Phi$. $\mathcal{H}_{\partial_{z}}(H_\xi^{-1}(p_y)/\Phi)$ is unitarily related to $B_1$ by the the arguments of the previous section.

We can now illustrate a form of ``quantisation commutes with reduction", in that there is a natural bijection between $G$-invariant sections of $\mathbb{L}\rightarrow T^*\mathbb{R}^2$ and sections of a new bundle $\mathbb{L}^\prime\rightarrow H_\xi^{-1}(p_y)/\Phi$. The $G$-invariant sections are constructed by lifting a $G$-action to $\mathbb{L}$ and averaging over its orbits. The lift of $\Phi$ to $\mathbb{L}$ is via the prequantisation $\xi\mapsto-i\hbar \nabla_{X_{\xi}}+H_\xi$. $X_\xi$ preserves the polarisation and may be integrated to a unitary $G$-action $\Phi^{\mathbb{L}}$ on $\mathcal{H}_{\partial_{\mathbf{z}}}$, giving
\begin{equation}\label{eqn:unitaryAction}
	\Phi^\mathbb{L}_\xi \psi(\mathbf{z}) = \psi\left(\Phi_\xi(\mathbf{z})\right)  \exp\left\{-\frac{i}{\hbar}\int_0^\xi( \theta(X_\xi)-H_\xi)d\xi^\prime\right\},
\end{equation}
though for $G$-invariance the local operator form is sufficient. A section $\psi\in\Gamma(\mathbb{L})$ is $G$-invariant if
\begin{equation*}
    \frac{d}{dt}\Bigg|_{t=0}\Phi^{\mathbb{L}}_{e^{t\xi}}(s(z))=0\,\,\,\,\text{for all }z\in T^*\mathbb{R}^2,
\end{equation*}
or, more simply, if $\nabla_{X_\xi}s=0$. This effectively imposes a lower-dimensional polarisation condition.
Recall the standard K\"ahler quantisation of $T^*\mathbb{R}^2$, we compute $X_\xi = -\xi(\partial_{z_y}+\partial_{\bar{z}_y})$ and so $G$-invariance implies
\begin{align*}
	\nabla_{X_\xi}\psi(\mathbf{z})&=\xi\left[-i\hbar\left(\nabla_{\partial_{z_y}}+\nabla_{\partial_{\bar{z}_y}}\right)+\frac{i}{2}(z_y-\bar{z}_y)\right]\psi(\mathbf{z})\\
	&=\xi\left[-i\hbar\frac{\partial F}{\partial z_y}+\frac{i}{2}z_yF\right]e^{-\frac{1}{4\hbar}(|\mathbf{z}_x|^2+|\mathbf{z}_y|^2)}\\
    &=0.
\end{align*}
Giving the condition
\begin{equation*}
\frac{\partial F}{\partial z_y}= \frac{1}{2\hbar}z_yF.
\end{equation*}
The set of $G$-invariant elements $\psi^\xi\in\mathcal{H}_{\partial_{\mathbf{z}}}$, denoted $\mathcal{H}^\xi_{\partial_{\mathbf{z}}}$, which we assume to have inherited a Hilbert space structure, take the form
\begin{equation*}
\psi^\xi(\mathbf{z}) = F(z_x,0)\exp\left\{-\frac{1}{4\hbar}\left(|z_x|^2+2iq_yp_y\right)\right\}.
\end{equation*}
The $\psi^\xi$, being fixed points of $\Phi^\mathbb{L}_\xi$, represent nontrivial equivalence classes $[\psi]_\xi$ of elements of $\mathcal{H}_{\partial_{\mathbf{z}}}$. We can choose as a representative of each class the section with vanishing dependence on $q_y$, i.e. we take
\begin{align*}
    [\psi^\xi(q_x,p_x,q_y,p_y)]_\xi &= \psi(q_x,p_x,[q_y]_\xi,p_y)\\
	&=\psi(q_x,p_x,[0]_\xi,p_y)\\
	&=\left[F(z_x,0)e^{-\frac{1}{4\hbar}|z_x|^2}\right]_\xi.
\end{align*}

There is a canonical bijective map between the space of $G$-invariant sections and sections of the reduced quantisation given by
\begin{equation}\label{eqn:basicReduction}
    \widehat{\pi}^\prime: \mathcal{H}^\xi_{\partial_{\mathbf{z}}}\longrightarrow \mathcal{H}_{\partial_{z}}(H_\xi^{-1}(p_y)/\Phi): \left[\psi^\xi(z_x,0)\right]_\xi\longmapsto \psi(z_x)
\end{equation}
with the inverse being the inclusion of the $\psi(z_x)$ into $\mathcal{H}^\xi_{\partial_{\mathbf{z}}}$.
Following \cite{hall2007unitarity} we see that to make such a bijection into a unitary map between the Hilbert spaces we need to weight it by the ``volume" of each $G$-orbit. Since the orbits in this case are straight lines, all have the same volume and the bijection is exactly a multiple of a unitary map $\widehat{\pi}^\prime = \lambda\widehat{\pi}$ where
\begin{eqnarray}
	\lambda = \int e^{-\frac{1}{4\hbar}|z|^2}dz\wedge d\bar{z}=\frac{4\pi}{\hbar}.
\end{eqnarray}
If all the orbits were not of the same volume, the construction of a unitary map would need to be sensitive to the exact value of the invariant momenta, which is problematic. Since our construction does not depend on the exact value of $p_y$ used to generate the reduced phase space, we can say that we have proved \textit{quantisation commutes with factorising harmonic oscillators} as a simple example of quantisation commuting with reduction. In each reduced Hilbert space, we may construct the unitary map into $B_1$, in particular, we could build $B_2$ from the tensor product of two reduced Hilbert spaces obtained via the reduction of $T^*\mathbb{R}^2$ in the $x$ and $y$ components.

\section{Reduced Quantisations of the 2DSHO}\label{sec:redQuant}
We are interested in solutions to \eqref{eqn:pwe} with an overall Gaussian profile, modulated by some special functions. Gaussianity is controlled by the presence of \eqref{eqn:2dsho} to select eigenfunction \cite{dennis2017swings}, but one eigenvalue is not sufficient to uniquely determine the eigenfunctions of a 2D system. The second operator generating the other eigenvalue usually corresponds to the separation constant of the coordinate system used. In this section, we use the polarisation of a quantisation as a way of carrying information about this observable, even when it does not correspond to a separation constant.

\subsection{The Reduction $T^*\mathbb{R}^2\rightarrow S^2$}
Recall that 2DSHO is defined by the Hamiltonian 
\begin{equation}\label{eqn:2DSHOham}
    H=\frac{1}{2}\left(p_x^2+p_y^2+q_x^2+q_y^2\right),
\end{equation}
on $T^*\mathbb{R}^2$ with the standard symplectic structure. 
Let $\Phi_t$ be the Hamiltonian $S^1$ action generated by $X_H$.
Alongside $H$ there are three other important invariants, the $\mathfrak{su}(2)$ algebra of constants of the motion generated by
\begin{align}\label{eqn:structreConstants}
	\tau_1& =q_xq_y+p_xp_y,\\
	\tau_2& =\frac{1}{2}(p_x^2+q_x^2-p_y^2-q_y^2) ,\\
	\tau_3& = xp_y-yp_x,
\end{align}
satisfying
\begin{equation}\label{eqn:structreConstants}
	\left\{\tau_i,\tau_j \right\} = \epsilon_{ijk}2\tau_k.
\end{equation}
Taken together, these invariants form the components of the Fradkin Tensor, the oscillator analogue of the Laplace–Runge–Lenz vector \cite{higgs1979dynamical}. We will now construct the symplectic reduction of $T^*\mathbb{R}^2$ using $X_H$ as the generator of the symmetry, this construction is standard in the literature (See 4.3.4 of \cite{abraham1978foundations}, \cite{cushman1982reduction} or section 4 of \cite{holm2011geometric} for different flavours), we provide it here for completeness. Let $t\in S^1$ be the exponentiation $t=e^{\lambda\xi}$ for $\lambda\in\mathbb{R}$ and $\xi\in \mathrm{Lie}(S^1)$. Explicitly, the generator is
\begin{equation}\label{eqn:generator}
\begin{split}
    \frac{d}{d\lambda}\Big|_{\lambda=0}\Phi_{e^{\lambda\xi}}(\mathbf{q},\mathbf{p}) &= \frac{d}{d\lambda}\Big|_{\lambda=0}(\mathbf{q}\cos(\lambda\xi)-\mathbf{p}\sin(\lambda\xi), \mathbf{p}\cos(\lambda\xi)+\mathbf{q}\sin(\lambda\xi))\\
    &=\xi\left(\mathbf{p}\partial_{\mathbf{q}}-\mathbf{q}\partial_{\mathbf{p}}\right)\\
    &=X_{\xi H}
    \end{split}
\end{equation}
This is by construction a Hamiltonian vector field with Hamiltonian $\xi H$. A moment map $\mu:T^*\mathbb{R}^2\rightarrow \mathbb{R}$ for $\Phi$ is given by
\begin{equation*}
    \xi H(\mathbf{q},\mathbf{p}) = \langle\mu(\mathbf{q},\mathbf{p}),\xi\rangle.
\end{equation*}
Take $\xi$ to be the unit element under the inner product, then $\mu(\mathbf{q},\mathbf{p}) = H(\mathbf{q},\mathbf{p})\,\xi^*$ -- the moment map scales $\xi^*$ by the Hamiltonian. Since $\left(\Phi_t\right)^*H(\mathbf{q},\mathbf{p})$ and $\mathrm{Ad}^*_g H$ are trivial actions, $\mu$ is $G$-equivariant. The reduced phase space is given by the quotient of $H^{-1}(E)\cong S^3$ by $S^1$, which is the well-known Hopf construction with base space $S^2$. It is instructive, however, to verify explicitly that $\mu$ satisfies the criteria for reduction.

Given our moment $\mu$, we need to confirm that $S^1$ acts freely and properly on $\mu^{-1}(E)$. $S^1$ acts freely as the only fixed point under $H$ is the zero orbit, which is accessible only when $E=0$, and $S^1$ has no proper subgroups. Furthermore, $G$ acts properly as $G$ itself is compact.
Finally, $\mu^{-1}(E)$ is regular for $E> 0$ as $E$ is a constant across the orbit. Having satisfied all the technical requirements we can now proceed with the reduction for the regular orbits.

Consider the inclusion $S^3_E\hookrightarrow T^*\mathbb{R}^2$ defined by \eqref{eqn:2DSHOham} labelled by $E$. For each $v\in T_p\mu^{-1}(E)$ we need to find $\xi\in\mathbb{R}$ such that $d\mu(p)(v-X_\xi)=0$. Since $\mathbb{R}$ is abelian however, $[\xi_i,\xi_j]$ vanishes and we are left with $\omega_{[p]}([v_1],[v_2]) = \omega(v_1,v_2)$
 restricted to $v_i\in T_p\mu^{-1}(E)$.
 The unique symplectic form on the reduced phase space is given by the consistency condition $\pi_E^*\,\omega_E = \iota^*\omega$ for the canonical projections and inclusions
 \begin{align}
     \pi_E&:\mu^{-1}(E)\rightarrow \mu^{-1}(E)/S^1,\\
     \iota_E&: \mu^{-1}(E)/S^1 \hookrightarrow \mu^{-1}(E).
 \end{align}
 Given that the $\tau_i$ project unaffected, being constants of the motion, we can use them to express $\omega_E$. On $T^*\mathbb{R}^2$ we have the canonically conjugate coordinates $\tau_3$ and $\phi=\frac{1}{2}\arctan(\tau_1/\tau_2)$, since these coordinates descend to the reduced space, the pullback of the reduced area form must be proportional to $d \tau_3\wedge d\phi$. The $\tau_i$ taken with $\mu$ form local coordinates on $T^*\mathbb{R}^2$ in which $\omega$ may be expressed, the consistency condition then admits the exact identification
 \begin{equation}
     \omega_E = d\tau_3\wedge d\phi.
 \end{equation}
 Here, $\iota^*\omega$ is recovered from the relationship $E^2 = \tau_1^2+\tau_2^2+\tau_3^2$. This reduced space we denote by $R_E=(S^2,\omega_E)$ and refer to it as the reduced phase space of \eqref{eqn:2DSHOham}, we have thus completed the symplectic reduction
\begin{equation}\label{eqn:reduction}
(T^*\mathbb{R}^2,\omega)\longrightarrow (R_E,\omega_E). 
\end{equation}

\subsection{Quantisation From Action-Angle Foliation}
With $\omega_E$ now in hand, we can perform geometric quantisation. On $S^2$ we can choose canonical cylindrical coordinates about any axis defined by a normalised $T=(\mathbf{u}\cdot \boldsymbol \tau)/E $ for some unit vector $\mathbf{u}$ and its associated azimuth $\chi$ to express $\omega_E$. Each canonical angle provides a different natural polarisation. Expressing $\omega_E=E\,dT\wedge d\chi$ we can consider the quantisation with respect to every choice of direction $\mathbf{u}$, where the polarisation in each case is defined by the (almost) Lagrangian foliation of $R_E$ by level sets of $T$. The leaf space for such a polarisation is the quotient $S^2/S^1$. We can safely ignore the singularities in the quotient and treat the leaf space as an open interval, as such elliptic singularities do not contribute to the quantisation \cite{hamilton2010locally}. Given this, the canonical bundle over the leaf space $\mathcal{K}_T$ (with singularities removed) is trivial and its square root $\delta_T$ we take to have unit section $\sqrt{dT}$.

The integrality condition for a prequantum line bundle over $R_E$ is modified by the addition of the $\frac{1}{2}$-form bundle \cite{woodhouse1992geometric}. We forgo a rigorous derivation of the effect and instead apply the EBK quantisation rule \cite{brack2018semiclassical}, which in this case performs an equivalent modification of the Bohr-Sommerfeld condition by the addition of the Maslov index, qualitatively related to the number of classical turning points encountered in a trajectory \cite{mcduff2017introduction}. The corrected condition is 
\begin{equation}\label{eqn:BSsphere}
	\frac{1}{2\pi \hbar}\int_{R_E}\omega_E   \in \mathbb{Z}\quad\! \Rightarrow \quad\! E = \frac{(n+1)}{2}\hbar\quad\text{for}\quad n=0,1,2,3,\ldots
\end{equation}
giving a collection of prequantum line bundles $\mathbb{L}_{n}\rightarrow R_E$ where $c_1(\mathbb{L}_n)=n$.

Working in the adapted trivialisation $\theta =-\chi \,dT$, according to the polarisation condition
\begin{equation*}
	\nabla_{\partial_\chi}\psi =\frac{\partial \psi}{\partial \chi}=0,
\end{equation*}
$\widetilde{\mathcal{H}}_{\partial_{\chi}}$ is comprised of all square-integrable wavefunctions $\widetilde{\psi}$ of the form
\begin{equation*}
    \widetilde{\psi} = \psi(T)\otimes\sqrt{dT},
\end{equation*}
with respect to the inner product
\begin{eqnarray}
    (\widetilde{\psi}_1,\widetilde{\psi}_2)= \int_{-E}^E\overline{\psi_1}\psi_2\,dT.
\end{eqnarray}
The quantisation $T\mapsto\widetilde{T}$ of the observable $T$ gives the expected constant operator
\begin{equation*}
    \tilde{T}\tilde{\psi} = T\tilde{\psi}.
\end{equation*}
$\widetilde{T}$ can provide an eigenbasis of sections for $\widetilde{\mathcal{H}}_{\partial_\chi}$.
Such sections will only be supported on a finite number of leaves called the \textit{Bohr-Sommerfeld leaves}, which we find by passing all the leaves $T^{-1}(t)$ through the Bohr-Sommerfeld condition (modified to include the square root bundle \cite{woodhouse1992geometric}). We wish to compute
\begin{equation}\label{eqn:BSsphereT}
	\int_{T^{-1}(t)} \theta\in\mathbb{Z}\quad\! \Rightarrow \quad\! T = \frac{m\hbar}{2}\quad\text{for}\quad m\in\mathbb{Z}
\end{equation}
The condition (\ref{eqn:BSsphereT}) must be consistent with the physically accessible values of $T\in(-E,E)$, where again we have excluded the singular leaves $T=\pm E$ which do not contribute to the quantisation.  In this polarisation $R_E$ will always support $2N+1$ Bohr-Sommerfeld leaves, where the identification $c_1(\mathbb{L})=N$ follows from the curvature condition.

We have thus constructed a quantum Hilbert space, based on the canonical polarisation by action-angle variables on $R_E$ in the coordinate system $(T,\varphi)$. It has a $2N+1$ dimensional basis of eigenfunctions, agreeing with the dimension of the spin $N/2$ irreducible representations of $SU(2)$. In the next section, we will perform a similar calculation in local coordinates more amenable to explicit representation of the quantum operators. Note that the above calculation is independent of any particular choice of $T$, all give rise to a Hilbert space of the same dimension. Later we will see these spaces are unitarily related, a fact which, it must be emphasised, is not automatically true in geometric quantisation.

There are of course other canonical coordinates on $S^2$ whose momenta may be quantised \cite{tolstikhin1995hyperspherical}, but they give rise to more pathological singular foliations that require more general sheaf-theoretic methods to handle \cite{hamilton2010geometric, mir2021geometric}.

\subsection{K\"ahler Quantisation}\label{sec:kahler}
We will now look at the same quantisation using K\"ahler polarisations. The resulting Hilbert space $\widetilde{\mathcal{H}}_{\partial_z}$ provides the basis from which we build our geometric description of GG modes. Using the complex coordinates
\begin{equation}
        \tau_1 = E^2\frac{z+\bar{z}}{E^2+|z|^2}\,,\quad\tau_2=E^2\frac{i(z-\bar{z})}{E^2+|z|^2}\,,\quad\tau_3 = E\frac{|z|^2-E^2}{|z|^2+E^2}
\end{equation}
obtained via stereographic projection of $R_E$, we then rescale the canonical angle as $\varphi=2\chi=\arctan(\tau_1/\tau_2)$ to simplify the coming algebra. The canonical Poisson bracket gives $\{\tau_i,\tau_j\}=\epsilon_{ijk}\tau_k$ with respect to the symplectic structure
\begin{equation}
    \omega_E = \frac{2i E}{\left(1+|z|^2\right)^2}dz\wedge d\bar{z},
\end{equation}
computed via pullback. The relevant K\"ahler potential is
\begin{equation}
	\mathcal{K} = 2E \log(1+|z|^2).
\end{equation}

Just as before, holomorphically polarised sections take the form $\psi(z) = F(z)\exp(-\frac{1}{2\hbar}\mathcal{K})$ for $F$ some holomorphic function of $z$. The fibres of $\mathbb{L}$ also come with the norm
\begin{equation}\label{eqn:reducedKaherNorm}
	||\psi||^2 = E\int|F(z)|^2e^{-\frac{\mathcal{K}}{\hbar}}\frac{2iE^3}{(E^2+|z|^2)^2}dz\wedge d\bar{z}.
\end{equation}

There is a unique square root of the canonical bundle on $S^2$ (See 3.2 and $\mathrm{3}^\prime$ of \cite{atiyah1971riemann}). Unit section $\sqrt{dz}$ serves as a basis for $\delta_{\partial_z}$. A basis for $\widetilde{\mathcal{H}}_{\partial_z}$ can be found using eigenfunctions of $\widetilde{T}$. A good choice of basis for $\widetilde{\mathcal{H}}_{\partial_z}$ is
\begin{equation}\label{eqn:redBargmann}
    \widetilde{\psi}_{j,m}=z^{m+j}e^{-\frac{1}{2\hbar}\mathcal{K}}\otimes\sqrt{dz}\quad\text{for}\quad-j\leq m\leq j,
\end{equation}
as then the operators
\begin{equation*}
    \widetilde{T}\widetilde{\psi}_{j,m} = \hbar m \widetilde{\psi}_{j,m},
    \end{equation*}
    \begin{equation*}
\left(\widetilde{T}^2+\widetilde{x}^2+\widetilde{y}^2\right)\widetilde{\psi}_{j,m} = \hbar^2 j(j+1)\widetilde{\psi}_{j,m},
\end{equation*}
manifest angular momentum eigenvalues. Let $x$, $y$ be a choice of Cartesian coordinates in the $T=0$ plane, the raising and lowering operators generating a unitary basis can then be written

\begin{align}
    a^\dag\widetilde{\psi}_{j,m} &:= \left(\frac{1+2j}{j-m}\sqrt{\frac{1+j+m}{1+j-m}}\right)\,(\widetilde{x}-i\widetilde{y})\widetilde{\psi}_{j,m} = \widetilde{\psi}_{j,m+1}, \\
    a\widetilde{\psi}_{j,m} &:= \left(\frac{1}{(1+2j)\sqrt{(2+j-m)(j+m)}}\right)(\widetilde{x}+i\widetilde{y})\widetilde{\psi}_{j,m} = \widetilde{\psi}_{j,m-1}
\end{align}
It is clear to see, now that we have appropriately abstracted the algebra, that $\widetilde{\mathcal{H}}_{\partial_z}$ may be unitarily identified with a subspace of $B_1$.

\section{Applications}\label{sec:moderotate}
\subsection{Quantisation of Rotations}
Given a quantisable observable, we can write an expression for the unitary flow it generates using \eqref{eqn:unitaryAction}. For the case of flows generated by the $\tau_i$, these are orthogonal rotations of $R_E$ and lift to $SU(2)$ rotations upon quantisation. While such flows preserve the canonical action-angle and K\"ahler polarisations, they are not all diagonal in the basis of eigenfunctions picked out by $\widetilde{T}$. The splitting of some initial state under such a rotation corresponds to an inter-basis expansion in terms of a new basis associated to the quantisation of the observable generating the rotation. 

The easiest way to perform such rotations in $\widetilde{\mathcal{H}}_{\partial_\chi}(R_E)$ is by considering projections of the same flow on $\widetilde{\mathcal{H}}_{\partial_{\mathbf{z}}}(T^*\mathbb{R}^2)$. We need a map $\widetilde{\mathcal{H}}_{\partial_{\mathbf{z}}}(T^*\mathbb{R}^2)\rightarrow \widetilde{\mathcal{H}}_{\partial_{z}}(R_E)$ which commutes with the $S^1$ reduction \eqref{eqn:reduction}, analogous to \eqref{eqn:basicReduction}.
As the construction of $\widetilde{\mathcal{H}}_{\partial_\chi}(R_E)$ is independent of any particular choice of $T$, we can take $T= \tau_2$ without loss of generality.
In terms of the basis $|n_x,n_y\rangle\in \widetilde{\mathcal{H}}_{\partial_{\mathbf{z}}}(T^*\mathbb{R}^2)$ in \eqref{eqn:sBasis}, we form the operators
\begin{equation}
    J=\widetilde{\tau}_2\,,\quad	J_+=\widetilde{\tau}_1+i\widetilde{\tau}_3\quad 
	\text{and}\quad J_-=\widetilde{\tau}_1-i\widetilde{\tau}_3.
\end{equation}
Their action on $B_2$ takes the particularly simple form in the Schwinger representation
\begin{align}
    J|n_x,n_y\rangle =& \frac{1}{2}(n_x-n_y)|n_x,n_y\rangle\\
    J_+|n_x,n_y\rangle =& \sqrt{(n_y)(n_x+1)}|n_x+1,n_y-1\rangle\\
    J_-|n_x,n_y\rangle=& \sqrt{(n_x)(n_y+1)}|n_x-1,n_y+1\rangle
\end{align}
Defining the new variables $j = (n_x+n_y)/2$ and $m=(n_x-n_y)/2$ realises the standard angular momentum algebra.
\begin{align}
\widetilde{H}|j,m\rangle=&j(j+1)|j,m\rangle\\
    J|j,m\rangle =&m|j,m\rangle\\
    J_\pm|j,m\rangle =&\sqrt{(j\mp m)(j\pm m+1)}|j,m\pm1\rangle
\end{align}
thus \eqref{eqn:sBasis} rewritten in a the new unitary basis
\begin{equation}\label{eqn:relatedBasis}
|j,m\rangle=\frac{\left(a^\dag_x\right)^{j+m}\left(a^\dag_y\right)^{j-m}}{\sqrt{(j+m)!\,(j-m)!}}|0\rangle
=\frac{\left(z_x\right)^{j+m}\left(z_y\right)^{j-m}}{\sqrt{(j+m)!\,(j-m)!}}\end{equation}
$\widetilde{H}$ determines the values taken by $j$ in half-integer steps, with each subspace corresponding to a particular $j$ spanned by $-j\leq m\leq j$ in integer steps.

Fixing a value $E$ taken by $\mu$, the $\widetilde{\rho}\,^{\tau_2}$ invariant states, the eigenstates of $J$, span a $2j+1$ dimensional subspace of $\widetilde{\mathcal{H}}_{\partial_{\mathbf{z}}}(T^*\mathbb{R}^2)$. 
By the same argument as in \ref{sec:qviar}, this subspace may be unitarily identified with $\widetilde{\mathcal{H}}_{\partial_{z}}(R_E)$.

This relationship lets us now use the more ergonomic expressions for flows on $T^*\mathbb{R}^2$ instead of on $R_E$. Another point we wish to emphasise is that we may think about these unitary transformations of $\widetilde{\mathcal{H}}_{\partial_{\mathbf{z}}}(T^*\mathbb{R}^2)$, roughly speaking, as geometrically resulting from the rotation of the polarisations on the underlying reduced phase spaces. In particular, the new polarisation is described by level sets of an observable which is diagonal in the newly rotated basis. An explicit form for the inter-basis expansions is given by the standard representation theory of $SU(2)$.

\begin{theorem}
    Let $\widetilde{\mathcal{H}}_{\partial_z}(R_E)$ be the corrected quantum Hilbert space resulting from the K\"ahler quantisation of the reduced phase space of the 2D simple harmonic oscillator with respect to the canonical coordinates $d\tau_2\wedge d\varphi$, without loss of generality. The quantisation of an arbitrary rotation with respect to the Euler angles $2\boldsymbol{\vartheta}=(2\vartheta_1,2\vartheta_2,2\vartheta_3)$ is given by
    \begin{equation*}
        \widetilde{R}(\boldsymbol{\vartheta})|j,m\rangle=\sum_{m^\prime=-j}^j\mathscr{D}^j_{m^\prime,m}(2\boldsymbol{\vartheta})|j,m^\prime \rangle
    \end{equation*}
    where $\mathscr{D}$ is Wigner's $D$-matrix  \cite{sakurai2014modern}. Here ``rotation" refers to any action of a (combination of) $\widetilde{\rho}^{\tau_i}$.
\end{theorem}
\begin{proof}
    We construct the equivalent action on $\widetilde{\mathcal{H}}_{\partial_{\mathbf{z}}}(T^*\mathbb{R}^2)$, the evaluation of which becomes standard when viewed in the basis $|j,m\rangle$. Each $\widetilde{\tau}_i$ acting on $\widetilde{\mathcal{H}}_{\partial_{\mathbf{z}}}(T^*\mathbb{R}^2)$ generates a flow with respect to a parameter $\chi$
    \begin{equation*}
    \widetilde{\rho}^{\,\tau_i}_\chi|j,m\rangle = \frac{\left(\rho^{\tau_i}_\chi a^\dag_x\right)^{j+m}\left(\rho^{\tau_i}_\chi a^\dag_y\right)^{j-m}}{\sqrt{(j+m)!(j-m)!}}\sqrt{\left(\rho^{\tau_i}_\chi\right)^*dz_x \wedge dz_y}\,\,\,.
\end{equation*}
While only two orthogonal rotations in the Euler angle prescription are needed to represent a generic rotation, for completeness we compute all the flows. The action $(\rho^{\tau_i}_\chi)^*\sqrt{dz_x \wedge dz_y}$ is always trivial and may be disregarded. In stating these results we will also move freely between the bases in \eqref{eqn:relatedBasis}. The simplest flow to compute is that generated by $\tau_2$, as the basis \eqref{eqn:relatedBasis} is already diagonal, it immediately follows that
\begin{equation}
\widetilde{\rho}^{\,\tau_2}_\chi|j,m\rangle = e^{2im\chi}|j,m\rangle.
\end{equation}
Next, $\tau_3$ generates rotations between the complex coordinates
\begin{equation}
   \begin{split} \widetilde{\rho}^{\,\tau_3}_\chi|j,m\rangle &= \frac{\left(a^\dag_x\cos\chi-a^\dag_y\sin\chi\right)^{j+m}\left(a^\dag_y\cos\chi+a_x^\dag\sin\chi\right)^{j-m}}{\sqrt{(j+m)!(j-m)!}}\sqrt{dz_x \wedge dz_y}\\
    &=\sum_{m^\prime=-j}^jd^{j}_{m^\prime,m}(2\chi)|j,m^\prime \rangle,
    \end{split}
\end{equation}
expressible in terms of Wigner's small $d$-matrix $d^{(j)}_{m^\prime,m}(\chi)$, given by
\begin{equation}
\begin{split}
    d^j_{m^\prime,m}(\alpha)=\sum_k(-1)^{k-m+m^\prime}\frac{\sqrt{(j+m)!\,(j-m)!\,(j+m^\prime)!\,(j-m^\prime)!}}{(j+m-k)!\,k!\,(j-k-m^\prime)!\,(k-m+m^\prime)!}\\
    \times\cos\left(\frac{\alpha}{2}\right)^{2j-2k+m-m^\prime}\sin\left(\frac{\alpha}{2}\right)^{2k-m+m^\prime}.
    \end{split}
\end{equation}
See \cite{sakurai2014modern} for a detailed calculation. Finally, $\tau_1$ generates complex rotation of the $z_i$
\begin{equation}
\widetilde{\rho}\,^{\tau_1}_\chi|j,m\rangle =\frac{\left(a^\dag_x\cos\chi+i a_y^\dag\sin\chi\right)^{j+m}\left(a^\dag_y\cos\chi+i a_x^\dag\sin\chi\right)^{j-m}}{\sqrt{(j+m)!\,(j-m)!}}|0\rangle.
\end{equation}
We compose the two simplest rotations to construct an arbitrary rotation
\begin{align*}
    \widetilde{R}(\vartheta_1,\vartheta_2,\vartheta_3)|j,m\rangle &= \left(\widetilde{\rho}^{\,\tau_2}_{\vartheta_1}\circ \widetilde{\rho}^{\,\tau_3}_{\vartheta_2}\circ \widetilde{\rho}^{\,\tau_2}_{\vartheta_3} \right)|j,m\rangle\\
    &= \sum_{m^\prime=-j}^je^{2im\vartheta_3}\,d^{j}_{m^\prime,m}(2\vartheta_2)\,e^{2im\vartheta_1}|j,m^\prime \rangle\\
    &= \sum_{m^\prime=-j}^j\mathscr{D}^j_{m^\prime,m}(2\boldsymbol{\vartheta})|j,m^\prime \rangle
\end{align*}

\end{proof}

\subsection{Rotating the Modal Poincar\'e Sphere and the Generalised Hermite-Laguerre-Gauss Modes}

When the image of $\widetilde{R}|j,m\rangle$ is evaluated over the $\mathbb{R}^2$ as an element of $\widetilde{\mathcal{H}}_{\partial_{\mathbf{p}}}$, a connection to structured light emerges. The unitary map $\widetilde{\mathcal{H}}_{\partial_{\mathbf{p}}}\rightarrow \widetilde{\mathcal{H}}_{\partial_{\mathbf{z}}}(T^*\mathbb{R}^2)$ is the inverse Segal-Bargmann transformation, see \cite{woodhouse1992geometric} for a derivation in the context of a pairing map in geometric quantisation. In the re-scaled variables $z_i=-i\xi_i\sqrt{2\hbar}$ and $q_i=x_i\sqrt{\hbar}$ we obtain the standard expression of the transform \cite{bargmann1961hilbert,hall2013quantum}.
Denote the transform by $P$, it is well known that the monomial basis maps under $P$ to a basis of Gaussian-weighted Hermite polynomials
\begin{equation}
    (P|j,m\rangle)(q_x,q_y) = \frac{\lambda_j}{\sqrt{(j+m)!\,(j-m)!\,\pi}}\;H_{j+m}\left(\frac{q_x}{\sqrt{\hbar}}\right)\,H_{j-m}\left(\frac{q_y}{\sqrt{\hbar}}\right)e^{-\frac{1}{2\hbar}(q_x^2+q_y^2)},
\end{equation}
up to some constant $\lambda_j$.
In keeping with convention in the optics literature \cite{gbur2016singular}, we define the (un-normalised) \textit{Hermite-Gauss} modes as
\begin{equation}
    \mathscr{HG}_{j+m,j-m}(q_x,q_y) = (-i)^{m-j}\;H_{j+m}\left(\frac{q_x}{\sqrt{\hbar}}\right)\,H_{j-m}\left(\frac{q_y}{\sqrt{\hbar}}\right)e^{-\frac{1}{2\hbar}(q_x^2+q_y^2)}.
\end{equation}
These functions are eigenfunctions of $\widetilde{H}$ and $J$ and may be extended to solutions of \eqref{eqn:pwe} \cite{gbur2016singular}, $J$ remains identified with $\tau_2\in C^\infty(T^*\mathbb{R}^2)$. It is straightforward to show that $\tau_2$ is the separation constant appearing when solving \eqref{eqn:pwe} in Cartesian coordinates in the $z=0$ plane. Under the action of $\widetilde{R}$ however, $\widetilde{R}|j,m\rangle$ is an eigenfunction of $\widetilde{\mathcal{H}}$ and $\widetilde{R}\,J\widetilde{R}^{-1}$, where the latter is a linear combination of the $\widetilde{\tau}_i$ but still corresponding to the quantisation of a particular function of $T^*\mathbb{R}^2$. For special rotations $\widetilde{R}\,J\widetilde{R}^{-1}$ becomes $\widetilde{\tau}_3$, corresponding to angular momentum as the separation constant or $\cos(\vartheta)\widetilde{\tau}_2-\sin(\vartheta)\widetilde{\tau}_1$, the separation constant of a rotated Cartesian coordinate system. The general linear combination of $\tau_i$ does not correspond to the separation constant of any coordinate system. The GG modes \cite{abramochkin2004generalized,abramochkin2010generalized} are the eigenfunctions of quantisation of these mixed separation constants.

\begin{theorem}\label{thm:main}
Define the GG mode of Gaussian mode order $j$, index $m$ and angles $0\leq\vartheta_1\leq \pi$ and $0\leq\vartheta_2\leq 2\pi$, denoted $\mathcal{G}_{j,m}(\vartheta_1,\vartheta_2)$ as in \cite{abramochkin2004generalized}. The $\mathcal{G}_{j,m}$ satisfy
    \begin{equation*}
        \mathcal{G}_{j,m}(\vartheta_1,\vartheta_2;q_x,q_y) = \left(P\,\widetilde{R}(\vartheta_1,\vartheta_2,0)\,|j,m\rangle\right)(q_x,q_y)
    \end{equation*}
\end{theorem}
\begin{proof}
    After substituting in the explicit form of $\mathscr{HG}$ , the proof follows by comparison of the expression with equation (7.6) of \cite{dennis2019gaussian} and the substitution $\hbar=1/k$.
\end{proof}

Note that this rotation is path-dependent. In \cite{dennis2019gaussian}, the authors use the freedom of the unused $\vartheta_3$ rotation to compensate for the geometric phase encountered. GG modes with indices $j,m$ may thus be identified with points on $S^2$. This is exactly the modal Poincar\'e sphere \cite{dennis2019gaussian}, an analogy with the Poincar\'e sphere describing optical polarisation. This demonstrates that the $\mathcal{G}_{j,m}(\vartheta_1,\vartheta_2)$ for a fixed $\vartheta_1, \vartheta_2$ form the natural basis of the corrected quantum Hilbert space (in position basis) when the polarisation is given by level sets of $R\, \tau_2\, R^{-1}$.

\section{Conclusion}
That the underlying symmetries of Gaussian modes admit this geometric representation over $S^2$ has long been a central theme of work on structured light. We have clarified the mathematical framework necessary to show how the modal Poincar\'e sphere emerges naturally from the geometric quantisation of an underlying harmonic oscillator. While in retrospect much of this can be seen as a consequence of the representation theory of $SU(2)$, especially in light of Kirillov's \textit{Orbit Method} \cite{kirillov2004lectures}, this concept continues to flourish in the optical community. We have placed this model in a rich geometric and topological setting, where the main symmetries are directly manifest. In the future, we hope to extend this work to other families of structured light beams, corresponding to different reduced phase spaces with more complicated polarisations. Recently, variations on this model which account for additional optical degrees of freedom, for example, the optical hypersphere \cite{sugic2021particle}, have already emerged and may benefit from being incorporated into this analysis.

\subsubsection*{Acknowledgments}
The author thanks Mark Dennis for the inspiration that grew into this project, as well as for his supervision in its early stages. The author also thanks Lisa Jeffrey for valuable comments on an early manuscript.

\subsubsection*{Funding}
This work was supported by the EPSRC Centre for Doctoral
Training in Topological Design (EP/S02297X/1) and partially supported by a Mitacs Globalink Research Award (IT40628).

\printbibliography

\end{document}